\documentclass[11pt]{article}

\usepackage{lineno,hyperref}
\usepackage{amssymb,amsmath,amsthm}
\modulolinenumbers[5]
\usepackage{xcolor}

\usepackage{times}
\usepackage{amsmath,amsfonts,amssymb,latexsym,epsfig,a4}
\usepackage{color,epsf}
\usepackage{graphicx}

\usepackage{alltt}
\usepackage{url}

\newtheorem{theorem}{Theorem}

\newtheorem{corollary}{Corollary}

\newtheorem{definition}{Definition}
\newtheorem{example}{Example}

\newtheorem{lemma}{Lemma}

\newtheorem{proposition}{Proposition}
\newtheorem{remark}{Remark}

\begin{document}


\title{B\'ezier surfaces with prescribed diagonals}

\author{A. Arnal\footnote{ Email: \texttt{ana.arnal@uji.es}. ORCID: 0000-0002-3283-3379}
	\and
	J. Monterde\footnote{Email: \texttt{monterde@uv.es}.  } }

\maketitle
\begin{center}
	Institut de Matem\`atiques i Aplicacions de Castell\'o (IMAC) and  De\-par\-ta\-ment de
	Ma\-te\-m\`a\-ti\-ques, Universitat Jaume I,
	E-12071 Cas\-te\-ll\'on, Spain.\\
	Dep. de Matem\`atiques, Universitat de Val\`encia,
	Burjassot (Val\`encia), Spain
\end{center}




\begin{abstract}
	The affine space of all tensor product B\'ezier patches of degree $n\times n$ with prescribed main diagonal curves is determined. First, the pair of B\'ezier curves which can be diagonals of a B\'ezier patch is characterized. Besides prescribing the diagonal curves, other related problems are considered, those where boundary curves or tangent planes along boundary curves are also prescribed.
\end{abstract}




\section{Introduction}

One of the paradigmatic problems in the realm of computer design is to build a surface once its boundary is prescribed. In the case of B\'ezier surfaces, there is an extensive literature dealing with the problem of building a Bézier surface with a prescribed boundary.

These kinds of method are fundamentally based on seeking the minimization of some functional with a defined geometric meaning in the space of all Bézier surfaces with that border. Examples of these functionals include the area functional, which gives rise to a non-linear problem, or quadratic functionals, such as the Dirichlet functional or similar. All of these methods can be easily implemented so that the designer is immediately able to visualize the result of any change to the initial data. 

Functional minimization can be performed if all the boundary curves are given, so the surface shape is quite controlled. Nevertheless, this is not in general true for PDE surfaces,  where the number of boundary curves that can be prescribed depends on the degree of the PDE (see \cite{mo} and related references). For example, in order to obtain harmonic surfaces only two boundary curves can be prescribed so the other two get out of control. Hence small changes in the prescribed data could produce big changes in the shape of the two free boundary curves.

Therefore, although the natural way of controlling the shape of a surface is through its boundary, a way of increasing control could be to fix some curves on the surface. For this purpose, the main diagonal curves of a tensor-product Bézier surface seems to be the best choice. So, the goal here is to give a method that allow prescribing not only the boundary curves but the diagonal curves of the surface.

In fact, diagonal curves play an important role in engineering applications of surface modeling, such as architectural design or the design and optimization of quad meshes in architectural geometry \cite{JWSP}. However, very few studies have been published on constrained modeling of tensor-product surfaces with diagonal curves, (see \cite{zhu} and \cite{zhu22}).

Two situations where controlling the diagonal curves of B\'ezier patches seems to be fundamental can be found in the literature. The first is related to conversion from rectangular patches into a triangular scheme.

Based on the concept of diagonal curve, \cite{SO} and \cite{Kolcun} respectively proposed S-Patch and BS-Patch. The authors of paper \cite{SSK}  added some requirements for the main diagonals of a Hermite bicubic patch to be able to adapt to a triangular scheme. The requirements are related to the degree of the diagonal curves. Indeed, as stated in \cite{Kolcun}, conversions between quadrilateral and triangular meshes need the degree of the diagonal curves to match the degree of the boundary curves of the quadrilateral meshes. Moreover, the reverse adaptation, from triangular to rectangular patches, would require building rectangular patches where not only the boundaries but also the diagonal curves are prescribed.

The second situation where diagonal curves play a main role is in the design of surfaces with some geometric property. In papers \cite{shsp} and \cite{shi} the authors deal with constant mean curvature surfaces. Such kind of surfaces allow parameterizations with isothermal coordinate lines. That is to say, the two families of coordinate lines are mutually orthogonal and tangent vectors at an intersection point are of the same length. In addition to these two families of coordinate lines, diagonal curves are used not only for better aesthetics, but also to improve structural stability. The structures that can be built using these methods could be modified as follows: once the mesh of coordinate lines and the diagonal lines are obtained, they can be used as prescribed data to obtain rectangular B\'ezier patches with additional properties.  

A more general situation can be found in papers \cite{MB} and \cite{MB4}, where the authors discuss the problem of adding features to a free form surface by applying one or several user-defined surface curves, which are seen as editable parameters. These user-prescribed curves can be diagonal curves, but are not necessarily so.

\bigskip
Since the surface control net is related to the diagonal curves control points, (see \cite{HF}), the point is that, in order to be prescribed, the boundary control points and the control points of the diagonal curves in particular must meet certain conditions that will be shown in this paper. In \cite{zhu}, the authors considered a similar problem, and provided a method for generating a surface from free boundaries and diagonal curves given by the user, that is, without taking into account that boundaries and diagonals are related, so they have to meet some conditions. They then proposed a correction of this given information using Lagrange method to minimize the corresponding differences.

Let us recall that when we prescribe $\mathcal{C}^0$-boundary conditions, that is, when the boundary of a $n\times n$ B\'ezier surface is prescribed, then $4n$ boundary control points are fixed from a total of $(n+1)^2$. Equivalently, we can say that the affine space of all B\'ezier surfaces with the same boundary is parameterized by the $(n-1)^2$ interior control points. Analogously, when we prescribe $\mathcal{C}^1$-boundary conditions, that is to say, when the boundary and the tangent planes along it are prescribed,  then the $4n$ boundary control points and the $4n-8$ control points adjacent to the boundary are fixed. Equivalently, the affine space of all B\'ezier surfaces with the same boundary and tangent planes along it is parameterized by the $(n-3)^2$ interior control points.

In a similar fashion, this paper deals with the following problem: What is the affine space of all B\'ezier surfaces with the same diagonals, or with the same boundaries and diagonals? In other words, is it possible to determine which control points are fixed if only the diagonals or the diagonals and the boundaries are prescribed? Which subset of control points parameterizes the corresponding affine space?

We solve the following problems:	
\begin{itemize}	\item[a)] Find the conditions, in terms of their control points, that two Bézier curves must fulfill in order to be the diagonal curves of an indeterminate Bézier surface. Then, find the conditions, again in terms of these diagonal curves control points, for them to be the diagonal curves of a Bézier surface with a fixed border.
	\item[b)] Once this goal has been achieved, find the dimension of the space of all Bézier surfaces that have certain boundary and diagonal curves.  This dimension indicates nothing but the number of control points of the Bézier surface that are free once both the boundary and the diagonals have been prescribed.
\end{itemize}

Finally we have also considered the above problems for both the two prescribed diagonals and the prescribed $\mathcal{C}^1$-boundary conditions.

Let us remark that free control points of a Bézier surface with a prescribed boundary and diagonals could be determined by minimizing a functional, such as the Dirichlet functional or any similar one.

\section{Admissible pairs of main diagonal curves}

Let us consider a tensor product B\'ezier surface of degree $n\times n$, ${\bf x}:[0,1]\times[0,1]\to {\mathbb R}^3$, with control net ${\mathcal P } = \{P_{i,j}\}_{i,j=0}^n$. The pair of main diagonal curves of this surface are the B\'ezier curves of degree $2n$:
$$\begin{array}{rcl}
	\alpha_1(t) &=& {\bf x}(t,t),\qquad\quad t\in[0,1],\\[2mm]
	\alpha_2(t) &=& {\bf x}(t,1-t),\ \quad t\in[0,1].
\end{array}
$$

Initially, we are  interested in the following problem: What are the conditions that a pair of degree $2n$ B\'ezier curves must meet in order to be the main diagonals of a tensor product B\'ezier surface?

A first condition can be obtained easily on the control polygons, as illustrated by Figure \ref{fig1}, since the midpoints have to coincide
\begin{equation}\label{unmig}\alpha_1\left(\frac12\right) = \alpha_2\left(\frac12\right).
\end{equation}
\begin{figure}[h]\label{punt-central-diagonals}
	\begin{center}
		\includegraphics[width=6cm]{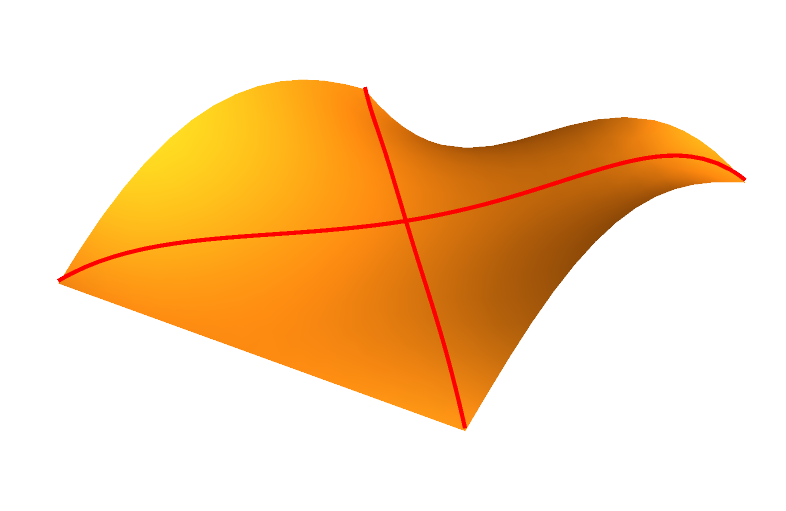}
		\caption{The two main diagonals of a tensor product B\'ezier surface. Notice that both have the same midpoint.}
		\label{fig1}
	\end{center}
\end{figure}
This simple fact has been pointed out in several references, see for instance \cite{Fa}, page 183, or \cite{zhu}. The novelty here is that such a condition will be completely expressed as a condition on the control points of both B\'ezier curves.

\begin{lemma}\label{lemma-1}
	The midpoint of two B\'ezier curves of degree $2n$ with control polygons
	$$
	{\mathcal Q} = \{Q_i\}_{i=0}^{2n},\qquad
	{\mathcal R} = \{R_i\}_{i=0}^{2n},
	$$
	is the same if and only if
	\begin{equation}\label{suma-condicions}
		\sum_{i=0}^{2n} \binom{2n}{i}Q_i = \sum_{i=0}^{2n} \binom{2n}{i} R_i.
	\end{equation}
\end{lemma}
\begin{proof} Since $$B^{2n}_i\left(\frac12\right) = \binom{2n}{i} \left(\frac12\right)^i \left(1-\frac12\right)^{2n-i}= \binom{2n}{i} \frac1{2^{2n}},$$
	then
	$$\begin{array}{rcl}
		\alpha_1\left(\frac12\right) &=& \frac1{2^{2n}}\sum_{i=0}^{2n} \binom{2n}{i}Q_i,\\[3mm]
		\alpha_2\left(\frac12\right) &=& \frac1{2^{2n}}\sum_{i=0}^{2n} \binom{2n}{i}R_i,
	\end{array}
	$$
	and the result is obtained.
\end{proof}

In paper \cite{HF} the authors state that the control points of a diagonal curve verify a particular dependence on the points of the Bézier surface control net.

\begin{lemma}\label{lemma-2} (See \cite{HF}) If $\{Q_i\}_{i=0}^{2n}$ is the control polygon of the first diagonal curve of a B\'ezier surface with control net $\{P_{i,j}\}_{i,j=0}^n$, the control points, $Q_k$, of the diagonal curve in terms of the control net are
	$$Q_k=\frac{1}{\binom{2n}{k}}\sum_{i=0}^k\binom{n}{i}\binom{n}{k-i} P_{i,k-i}\quad \text{for} \quad  k=0,\dots,2n.$$
\end{lemma}
The fact is that $Q_k$ depends only on the control points $P_{i,j}$ such that $i+j=k$. Let us consider only those control points, $Q_k$, with $k\leq n$. The others follow by symmetry. Treat the $P_{i,j}$ with $i+j=k$ as a degree $k$ Bézier curve and degree-elevate this curve $2n-k$ times, that is, until it is of degree $2n$. Then, the middle control point on this degree-elevated curve equals $Q_k$, (see \cite{HF} and Figure \ref{farinfig}).
\begin{figure}[h!]
	\centering
	\includegraphics[width =4.5cm]{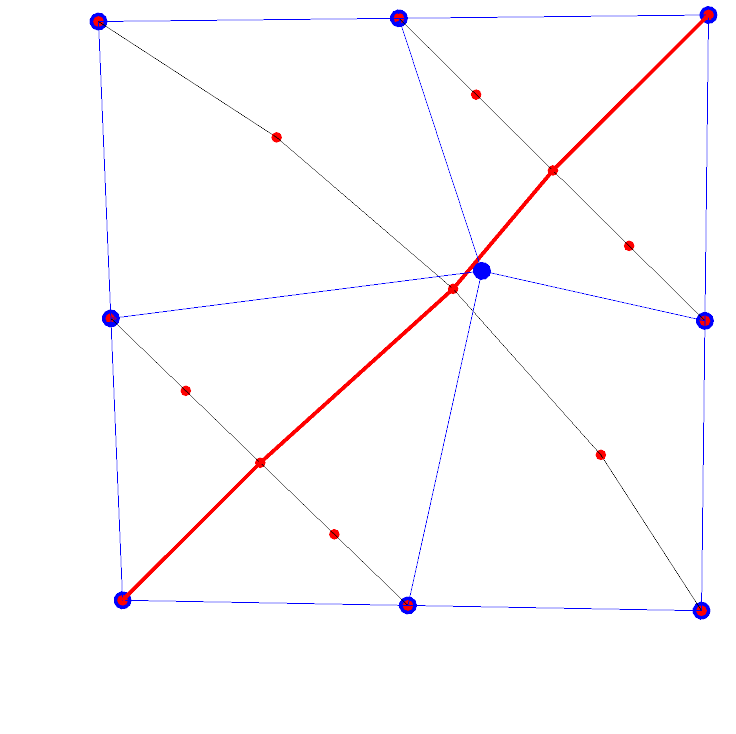}\qquad\qquad\includegraphics[width =6
	cm]{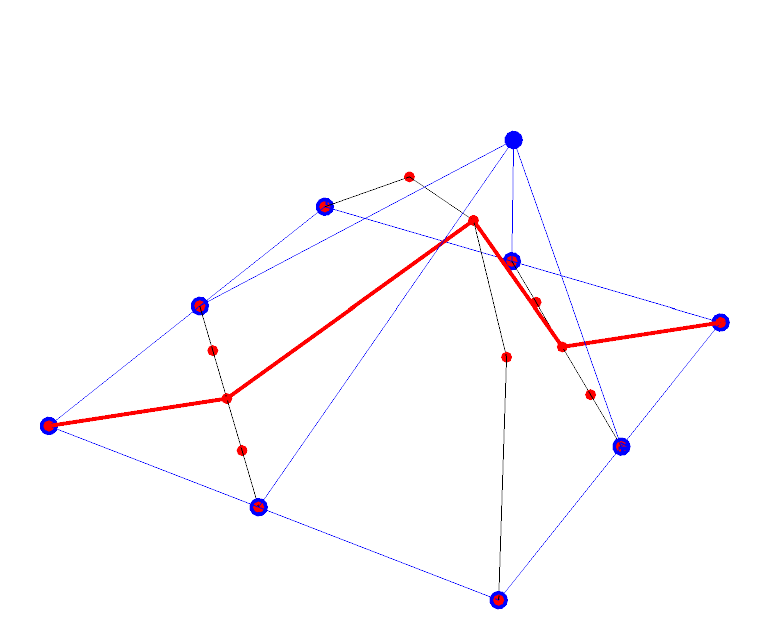}
	\caption{Control polygon (red) of the diagonal curve of a biquadratic Bézier surface obtained by repeated degree elevation. The corresponding degree elevation control points are also in red.}\label{farinfig}
\end{figure}

As we said, we are interested in finding the conditions that two Bézier curves must meet in order to be the main diagonals of some Bézier surface. Obviously, these curves must have the same midpoint, but, as we will see later in Theorem \ref{two-necessary-conditions}, condition (\ref{suma-condicions}) is necessary but not sufficient for two curves to be the diagonal curves of a Bézier surface. In fact, this condition splits in two more conditions.

As we will show, this splitting is a consequence of the fact that any diagonal control point, $Q_k$, depends only on the control points in a diagonal line of the control net, $P_{i,j}$ with $i+j=k$.  Thus, condition (\ref{suma-condicions}) splits into two sufficient conditions according to the parity of the indexes of the control points. We will state this in Theorem \ref{two-necessary-conditions}, but beforehand we will need the following auxiliary Lemma.

\begin{lemma}\label{lemma-3}
	Let us suppose that $\{Q_i\}_{i=0}^{2n}$ is the control polygon of the first diagonal curve, ${\bf x}(t,t)$, and that $\{R_i\}_{i=0}^{2n}$ of the second diagonal curve, ${\bf x}(t,1-t)$, of a B\'ezier surface with control net $\{P_{i,j}\}_{i,j=0}^n$, then  $\{(-1)^iQ_i\}_{i=0}^{2n}$ and $\{(-1)^{n-i}R_i\}_{i=0}^{2n}$ are the control points of the diagonal curves of a B\'ezier surface with control net $\{(-1)^{i+j}P_{i,j}\}_{i,j=0}^n$.
\end{lemma}

\begin{proof}
	As we said before, if $\{Q_i\}_{i=0}^{2n}$ is the control polygon of the first diagonal curve of a B\'ezier surface with control net $\{P_{i,j}\}_{i,j=0}^n$ then
	$$x(t,t)=\sum_{k=0}^{2n}\frac{B^{2n}_k(t)}{\binom{2n}{k}} \left(\sum_{i=0}^k\binom{n}{i}\binom{n}{k-i} P_{i,k-i} \right) = \sum_{k=0}^{2n}B^{2n}_k(t) Q_k,$$
	and if we multiply this equation by $(-1)^k$ then
	$$\sum_{k=0}^{2n}\frac{B^{2n}_k(t)}{\binom{2n}{k} } \left(\sum_{i=0}^k\binom{n}{i}\binom{n}{k-i} (-1)^{i+k-i}P_{i,k-i} \right) = \sum_{k=0}^{2n}B^{2n}_k(t) (-1)^kQ_k,$$
	and the first statement is proved.
	\bigskip
	The procedure is analogous for the second statement.
	
\end{proof}

Let us recall again the relation between the control net and the control points of the diagonal curves:
$$\aligned x(t,t)&=\sum_{k=0}^{2n}\frac{B^{2n}_k(t)}{\binom{2n}k}  \left(\sum_{i=0}^k\binom{n}{i}\binom{n}{k-i} P_{i,k-i} \right) = \sum_{k=0}^{2n}B^{2n}_k(t) Q_k,\\
x(t,1-t)&=\sum_{k=0}^{2n}\frac{B^{2n}_k(t)}{\binom{2n}k}  \left(\sum_{i=0}^k\binom{n}{i}\binom{n}{k-i} P_{i,n-k+i} \right) = \sum_{k=0}^{2n}B^{2n}_k(t) R_k\endaligned$$
and then denote
\begin{equation}\label{DE}
	\left\{
	\begin{array}{rcl}
		D^n_k({\mathcal P } ) &:=& \sum_{i=0}^{k} \binom{n}{i} \binom{n}{k-i} P_{i,k-i},\\
		E^n_k({\mathcal P } ) &:=& \sum_{i=0}^{k} \binom{n}{i} \binom{n}{k-i} P_{i,n-k+i},
	\end{array}
	\right.
\end{equation}
where \begin{equation}\label{QR_P}Q_k=\frac{D^n_k({\mathcal P } )}{\binom{2n}{k}}\quad\quad
R_k=\frac{E^n_k({\mathcal P } )}{\binom{2n}{k}}.\end{equation}

Now we will give the necessary and sufficient conditions for two control polygons to be the control points of the two diagonal curves.
\begin{theorem}\label{two-necessary-conditions}
	Two B\'ezier curves of degree $2n$ with control polygons
	$$
	{\mathcal Q} = \{Q_i\}_{i=0}^{2n},\qquad
	{\mathcal R} = \{R_i\}_{i=0}^{2n},
	$$
	are the main diagonals of a tensor product B\'ezier surface of degree $n\times n$ if and only if, for even $n$
	\begin{equation}\label{condicions-even}
		\begin{array}{rcl}
			\sum_{i=0}^{n} \binom{2n}{2i}Q_{2i}&=&\sum_{i=0}^{n} \binom{2n}{2i}R_{2i},\\[3mm]
			\sum_{i=0}^{n-1} \binom{2n}{2i+1}Q_{2i+1} &=&\sum_{i=0}^{n-1} \binom{2n}{2i+1}R_{2i+1},
		\end{array}
	\end{equation}
	or, for odd $n$
	\begin{equation}\label{condicions-odd}
		\begin{array}{rcl}
			\sum_{i=0}^{n} \binom{2n}{2i}Q_{2i}&=&\sum_{i=0}^{n-1} \binom{2n}{2i+1}R_{2i+1},\\[3mm]
			\sum_{i=0}^{n-1} \binom{2n}{2i+1}Q_{2i+1} &=&\sum_{i=0}^{n} \binom{2n}{2i}R_{2i}.
		\end{array}
	\end{equation}
\end{theorem}

\begin{proof}
	Let us suppose that $n$ is even. The proof is analogous if $ n$ is odd.
	
	First we will prove that if the curves are the main diagonals of a tensor product B\'ezier surface with control net $\{P_{i,j}\}_{i,j=0}^{n}$ Eq. (\ref{condicions-even}) is fulfilled. By Lemma \ref{lemma-1} we have 
	$$\sum_{i=0}^{2n} \binom{2n}{i}Q_i = \sum_{i=0}^{2n} \binom{2n}{i}R_i,$$
	which can be rewritten as
	\begin{equation}\label{even-1}\sum_{i=0}^{n} \binom{2n}{2i}Q_{2i}+\sum_{i=0}^{n-1} \binom{2n}{2i+1}Q_{2i+1} =\sum_{i=0}^{n} \binom{2n}{2i}R_{2i}+\sum_{i=0}^{n-1} \binom{2n}{2i+1}R_{2i+1}.
	\end{equation}
	
	Let us consider the Bézier surface with control net $\{(-1)^{i+j}P_{i,j}\}_{i,j=0}^{n}$. By Lemma \ref{lemma-3} the control polygons of the main diagonals are $\{(-1)^iQ_i\}_{i=0}^{2n}$ and $\{(-1)^{n-i}R_i\}_{i=0}^{2n}$, then if we apply Lemma \ref{lemma-1} again we have 
	\begin{equation}\label{even-2}\sum_{i=0}^{n} \binom{2n}{2i}Q_{2i}-\sum_{i=0}^{n-1} \binom{2n}{2i+1}Q_{2i+1} =\sum_{i=0}^{n} \binom{2n}{2i}R_{2i}-\sum_{i=0}^{n-1} \binom{2n}{2i+1}R_{2i+1}.
	\end{equation}
	Adding and substracting the equations (\ref{even-1}) and (\ref{even-2}), we get Eq.\ref{condicions-even}.
	
	In order to prove the converse we need to check that if Eq. (\ref{condicions-even}) is fulfilled then the linear system
	\begin{equation}\label{the-system}
		\left\{
		\begin{array}{rcl}
			D^n_k({\mathcal P } )&=&  \binom{2n}{k} Q_k, \qquad k=0,\dots, 2n\\[2mm]
			E^n_k({\mathcal P } ) &=&  \binom{2n}{k} R_k, \qquad k=0,\dots, 2n
		\end{array},
		\right.
	\end{equation}
	with the surface control points as unknown variables, has a solution. The conditions in the statement are what we need for the linear system, Eq. (\ref{the-system}), to be compatible.
	It is easy to check that
	\begin{equation}\label{condicions-even-DE}\left\{
		\begin{array}{lcr}
			\sum_{i=0}^{n}	D^n_{2i}({\mathcal P } ) - \sum_{i=0}^{n}	E^n_{2i}({\mathcal P } )&=&0\\[2mm]
			\sum_{i=0}^{n-1}	D^n_{2i+1}({\mathcal P } ) - \sum_{i=0}^{n-1}	E^n_{2i+1}({\mathcal P } )&=&0
		\end{array}
		\right..
	\end{equation}
	
	We will show, in addition, that these are the only  two vanishing linear combinations of the left hand side of Eq. (\ref{the-system}). Once this is proved then conditions in Eq. (\ref{condicions-even}), that are same linear combination but of the right hand side terms of system (\ref{the-system}), will imply the compatibility of the system.

	So, let us consider an arbitrary vanishing linear combination of the left hand side in system (\ref{the-system})
	\begin{equation}\label{prs}\sum_{i=0}^{n} \lambda_{i} D^n_{i}({\mathcal P })+\sum_{i=0}^{n} \mu_{i} E^n_{i}({\mathcal P } )=0.
	\end{equation}
	
	Any control point $P_{r,s}$ appears only once in $D^n_{r+s}({\mathcal P })$ and in $E^n_{n+r-s}({\mathcal P })$ and with the same coefficient $\binom{n}{r}\binom{n}{s}$, see Eq. (\ref{DE}). Then, the coefficient of $P_{r,s}$ in Eq. (\ref{prs}) is
	$$  \binom{n}{r}\binom{n}{s}(\lambda_{r+s} + \mu_{n+r-s}) ,
	$$
	Then, since the linear combination in Eq. (\ref{prs}) must vanish for all $P_{r,s}$ in the control net,  $\lambda_{r+s}+ \mu_{n+r-s}=0$, the linear combination coefficients are  opposites, as in Eq. (\ref{condicions-even-DE}).
	This happens for all sums $r+s$, so we have $\lambda_{r+s}=\lambda_{r-i+s+i}=-\mu_{n+r-s+2i}$, in other words, all $\lambda_{r+s}$ coefficients with the same parity of subindex are the opposite of all $\mu_{r+s}$ coefficients with the same parity. Thus, for all even (odd) subindices $\lambda_{r+s}=\lambda_{even \,(odd)}=-\mu_{r+s}$. Thus
	$$\begin{array}{rcl}
		\lambda_{even}\,(\sum_{i=0}^{n}	D^n_{2i}({\mathcal P } ) - \sum_{i=0}^{n}	E^n_{2i}({\mathcal P } ))+
		\lambda_{odd}\,(\sum_{i=0}^{n-1}	D^n_{2i+1}({\mathcal P } ) - \sum_{i=0}^{n-1}	 E^n_{2i+1}({\mathcal P } ))=0.
	\end{array}
	$$
	
	Therefore, this linear combination of equations in system (\ref{the-system}) is a linear combination of the equations in  (\ref{condicions-even-DE}). In other words, as mentioned before, Eq. (\ref{condicions-even-DE}) contains the two only vanishing linear combinations of the left hand side of Eq. (\ref{the-system}).

	Since the right hand side verifies the same relation, Eq. (\ref{condicions-even}), we can be sure that system (\ref{the-system}) has a solution, so a Bézier control net that fits the given diagonals control polygons does exist.

\end{proof}

Although we are mainly interested in dealing with both main diagonals simultaneously, it would be possible to deal with the case where only one diagonal is prescribed. Notice that although we will not consider it here, this case is simpler since the fundamental midpoint restriction, in Eq. (\ref{unmig}), does not apply.

\begin{remark}
	The proof of the splitting of Eq. (\ref{suma-condicions}) has been obtained thanks to auxiliary Lemma \ref{lemma-2}, but there is an alternative way to prove it using rational B\'ezier surfaces. As said in \cite{Fa},(page 183), the two diagonal curves of a rational B\'ezier surface, ${\bf x}(u,v)$, not only pass through point ${\bf x}(\frac12,\frac12)$, but also meet one more time: $\alpha_1(\infty) = \alpha_2(\infty)$.  This second condition can be written in terms of the (projective) control points of the two diagonals as in Eq. (\ref{suma-condicions}) but after a change of sign of the odd terms. When the rational B\'ezier surface reduces to a B\'ezier surface, in other words, when all the weights are the same, we obtain Eq. (\ref{even-2}).
\end{remark}

The relation between the control points of the diagonal curves, given in Theorem \ref{two-necessary-conditions}, makes it possible in the odd case to solve both central control points, $Q_n$ and $R_n$, in terms of the other control points. However, it is not possible to do the same in the even case since $Q_n$ and $R_n$ appear in the same equation. We illustrate this relation in the following corollary.

\begin{corollary}
	\label{corollary-two-necessary-conditions}
	${\mathcal Q} = \{Q_i\}_{i=0}^{2n}$ and ${\mathcal R} = \{R_i\}_{i=0}^{2n}$ are the control polygons of the main diagonals of a Bézier surface if and only if for even $n$
	\begin{equation}\label{condicions-even-solved}
		\begin{array}{rcl}
			Q_n &=&\frac1{\binom{2n}{n}}\left(\sum_{i=0}^{n} \binom{2n}{2i}R_{2i}-\sum_{i=0,2i\ne n}^{n} \binom{2n}{2i}Q_{2i}\right),\\[3mm]
			R_{n+1}&=&\frac1{\binom{2n}{n}}\left(\sum_{i=0}^{n-1} \binom{2n}{2i+1}Q_{2i+1} -\sum_{i=0,2i\ne n}^{n-1} \binom{2n}{2i+1}R_{2i+1}\right),
		\end{array}
	\end{equation}
	and,  for odd $n$
	\begin{equation}\label{condicions-odd-solved}
		\begin{array}{rcl}
			Q_n &=&\frac1{\binom{2n}{n}}\left( \sum_{i=0}^{n-1} \binom{2n}{2i+1}R_{2i+1}-\sum_{i=0, 2i\ne n}^{n} \binom{2n}{2i}Q_{2i}\right),\\[3mm]
			R_n&=&\frac1{\binom{2n}{n}}\left(\sum_{i=0}^{n-1} \binom{2n}{2i+1}Q_{2i+1} -\sum_{i=0,2i\ne n}^{n} \binom{2n}{2i}R_{2i}\right).
		\end{array}
	\end{equation}
\end{corollary}

At this point, let us remark that user could give two curves the  role of diagonal curves which do not intersect. In that case, the substitution of both central control points of the given diagonal curves by those computed with Eq. \ref{condicions-even-solved} in Corollary \ref{corollary-two-necessary-conditions}, for odd degre, would change the given input to a pair of admissible diagonal curves, see Figure \ref{newinput}. For even degree, the corresponding control point subtitutions should be done, Eq. (\ref{condicions-odd-solved}).

\begin{figure}[h!]
	\centering
	\includegraphics[width =6cm]{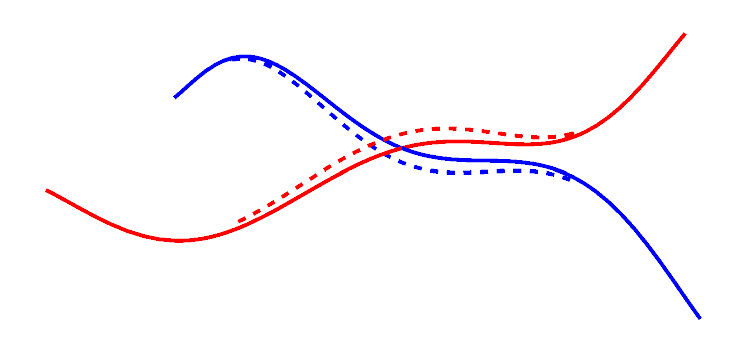}
	\caption{If a designer gave two curves that do not intersect (dashed line), one could modify them  slightly by means of Corollary \ref{corollary-two-necessary-conditions} in order to get two admissible diagonal curves.}\label{newinput}
\end{figure}

For odd degree, the control point which is modified is the central one, so the shape of the new diagonal curves is not too different to the initial one. In the case of even degree, if there was some symmetry in the given diagonals it could be lost in the modified diagonal curves. Nevertheless, a way of solving  this handicap easily would be to carry out a degree elevation of the even degree curves given by the user, and then proceed with the degree-elevated odd case.

An alternative way of adapting the information provided by the user was given in papers \cite{zhu} and \cite{zhu22}, where the authors minimized the distance between the given and new diagonal curves. Now, in the same fashion but considering the necessary conditions in Theorem \ref{two-necessary-conditions} (or Corollary \ref{corollary-two-necessary-conditions}), an analogous minimization procedure could be performed.

\begin{definition}
	Two tensor product B\'ezier surfaces of degree $n\times n$ are said to be diagonals equivalent if they have the same main diagonal curves.
\end{definition}

\begin{proposition}\label{prop1}
	The class of diagonals equivalent surfaces is an affine subspace that can be parameterized by $(n-1)^2$ control points.
\end{proposition}
\begin{proof} Let us compute the dimension of the affine subspace of solutions of system (\ref{the-system}). The number of variables, $\{P_{i,j}\}_{i,j=0}^n$, is $(n+1)^2$. The number of equations of the linear system  is $2(2n+1) = 4n+2$, but its rank is two units less, as we have just seen in the proof of Theorem \ref{two-necessary-conditions}. Therefore, the dimension is $(n+1)^2-(4n+2-2) = (n-1)^2$.
\end{proof}

\begin{example}	For $n=2$ the two conditions in Eq. (\ref{condicions-even}) are
	$$\left\{\begin{array}{rcl}
		Q_0+6Q_2+Q_4&=& R_0+6R_2+R_4,\\
		Q_1 +  Q_3 &=& R_1 + R_3.
	\end{array}
	\right.
	$$
	
	Therefore, if we solve it, as stated in Corollary \ref{corollary-two-necessary-conditions},
	$$\aligned &Q_2=\frac{1}{6} \left(R_0+6 R_2+R_4-Q_0-Q_4\right),\\
	&R_3= Q_1+Q_3-R_1,\endaligned
	$$
	then
	$$
	\begin{array}{rcl}
		{\mathcal Q} &=& \{Q_0,Q_1,\frac{1}{6} \left(R_0+6 R_2+R_4-Q_0-Q_4\right), Q_3,Q_4\}\\[3mm]
		{\mathcal R} &=& \{R_0,R_1,R_2, Q_1+Q_3-R_1,R_4\}
	\end{array}
	$$
	are the control polygons of two B\'ezier curves admissible as main diagonals of a degree $2\times2$ B\'ezier surface.	Indeed, any B\'ezier surface with a control net
	$$\left(
	\begin{array}{ccc}
		Q_0 & \bf{P_{0,1}} & R_0 \\
		2 Q_1-\bf{P_{0,1}} & \frac{1}{4} \left(6 R_2-Q_0-Q_4\right) & 2 R_1-{\bf P_{0,1}} \\
		R_4 & {\bf P_{0,1}}+2 Q_3-2 R_1 & Q_4 \\
	\end{array}
	\right)
	$$
	has as its main diagonals the B\'ezier curves with the control polygons ${\mathcal Q}$ and	${\mathcal R}$. Notice that once the diagonal curves are prescribed there is only one degree of freedom, $P_{0,1}$.
\end{example}

\begin{example} 
	For $n=3$ from Corollary \ref{corollary-two-necessary-conditions} we have
	$$\begin{array}{rcl}
		R_3 &=& \frac1{20} \left(Q_0 + 15 Q_2 + 15 Q_4 - 6R_1 - 6R_5\right),\\[3mm]
		Q_3&=& \frac1{20} \left(R_0+ 15R_2+15 R_4+ R_6-6Q_1 - 6Q_5\right),
	\end{array}
	$$
	then $
	{\mathcal Q} = \{Q_i\}_{i=0}^{6}$ and $	{\mathcal R} = \{R_i\}_{i=0}^{6},
	$
	are the control polygons of two B\'ezier curves that are admissible as main diagonals of a degree $3\times3$ B\'ezier surface with a control net
	$$\left(
	\begin{array}{cccc}
		Q_0 & \bf{P_{0,1}} & \bf{P_{0,2}} & R_0 \\
		2 Q_1-\bf{P_{0,1}} & \bf{P_{1,1}} & \bf{P_{1,2}} & 2 R_1-\bf{P_{0,2}} \\
		5 Q_2 	-{\bf P_{0,2}}-3 {\bf P_{1,1}}&P_{2,1}   &P_{2,2} & 5 R_2-{\bf P_{0,1}}-3 {\bf P_{1,2}} \\
		R_6 &P_{3,1}  & P_{3,2}& Q_6 \\
	\end{array}
	\right)$$
where	
$$\aligned& P_{2,1}	=\frac{5 }{3}R_2+\frac{5 }{3}R_4-\frac{2}{3} Q_1-\frac{2 }{3}Q_5-{\bf P_{1,2}}\\
&P_{2,2}=\frac{5 }{3}Q_2+\frac{5 }{3}Q_4-\frac{2 }{3}R_1-\frac{2 }{3}R_5 -{\bf P_{1,1}}\\
&P_{3,1}={\bf P_{0,2}}+3 {\bf P_{1,1}}-5 Q_2+2 R_5\\
&P_{3,2}= {\bf P_{0,1}}+3 {\bf P_{1,2}}+2 Q_5-5 R_2
\endaligned
$$

Notice that there are four degrees of freedom, $P_{0,1},P_{0,2}, P_{1,1}, P_{1,2} $, in other words, the space of all degree $3$ tensor product Bézier surfaces with given diagonals is an affine space of dimension $2^2$. In Figure \ref{fig}, with the same given diagonal curves, we plot three different surfaces that belong to this affine space.
\end{example}

\begin{center}
\begin{figure}[h!]\centering
	\qquad\qquad	\includegraphics[width =4.2cm]{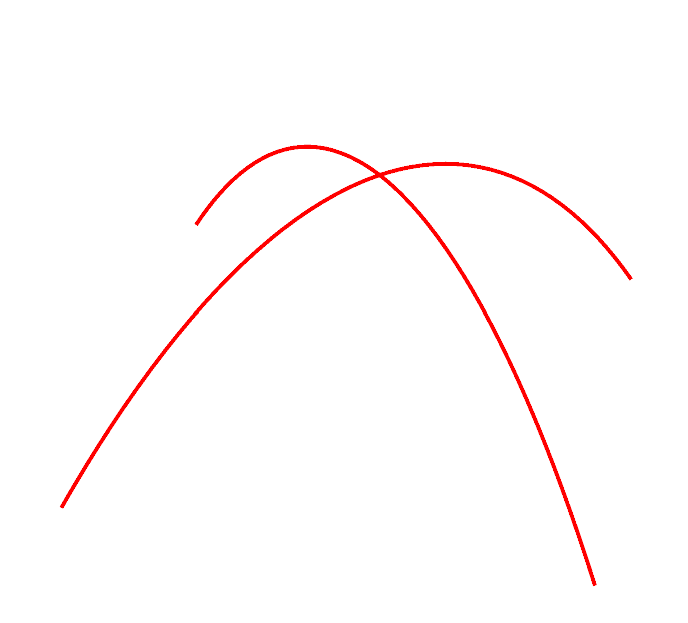}\qquad\qquad\qquad	\includegraphics[width =4.7cm]{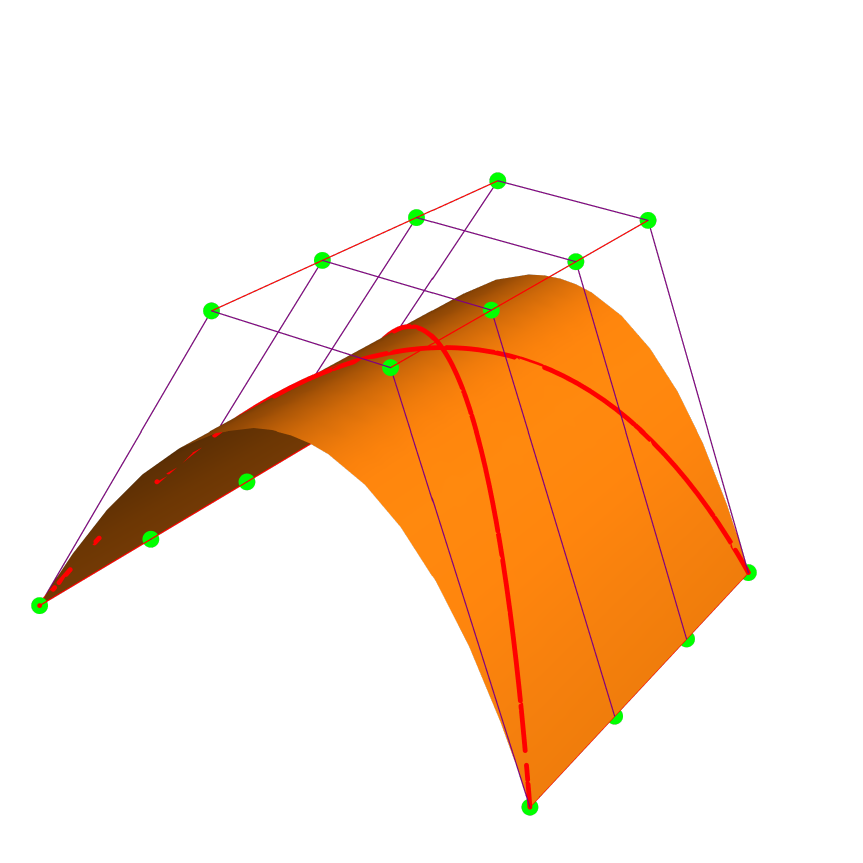}	\\[-1mm]
	\includegraphics[width =6.7cm]{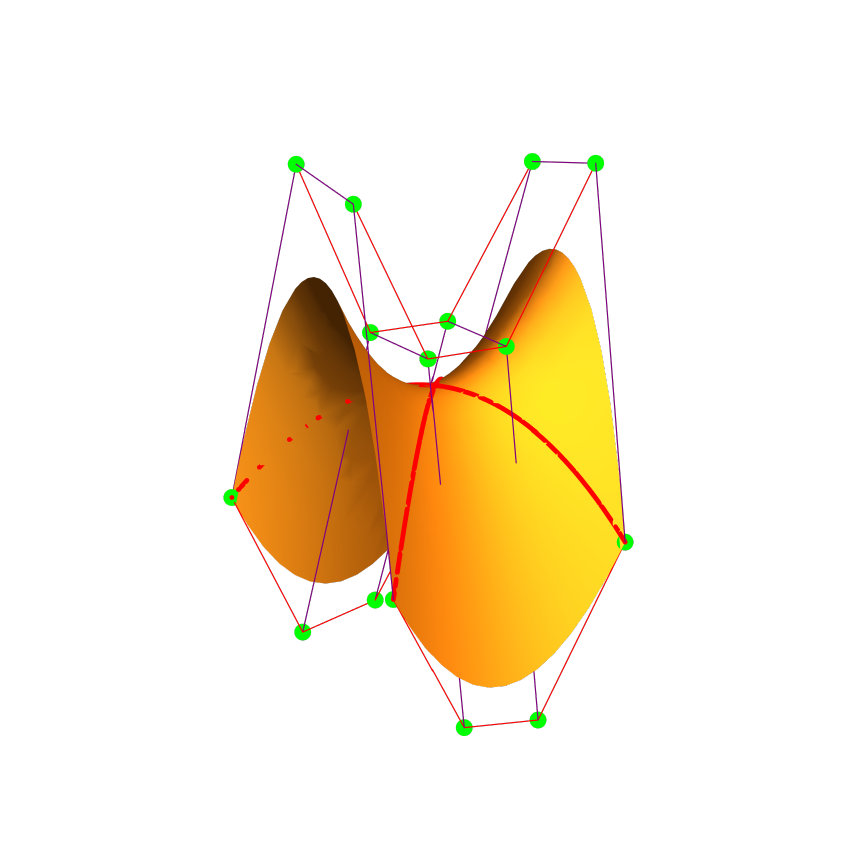}\, \includegraphics[width =6cm]{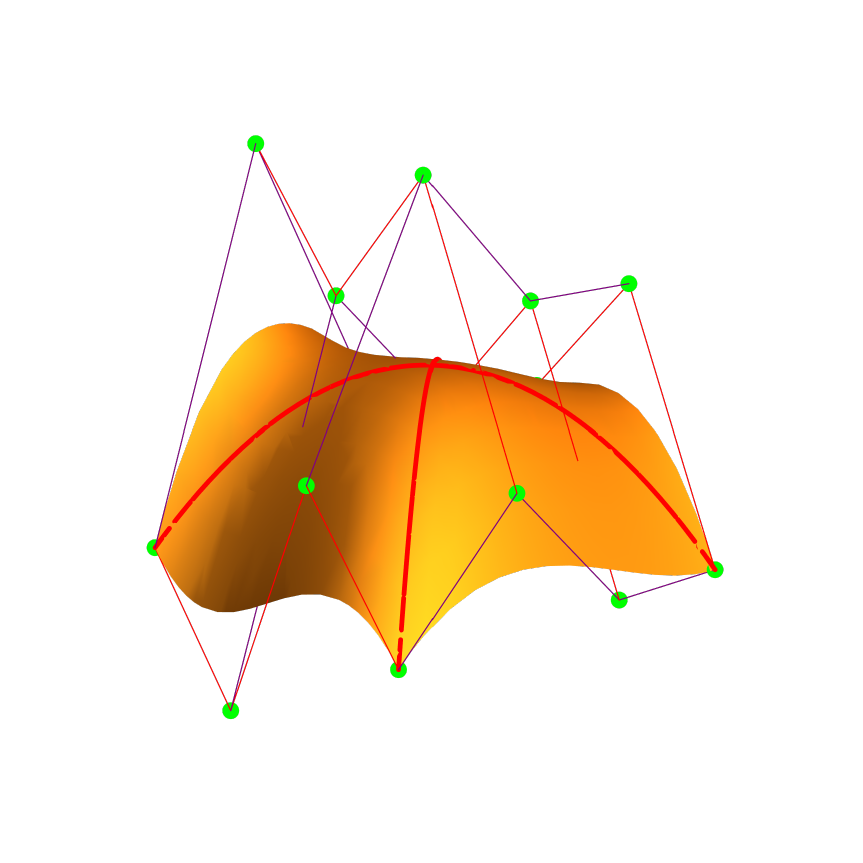}
	\caption{$n=3$. Three Bézier surfaces with the same prescribed main diagonals.}\label{fig}
\end{figure}
\end{center}
\section{Generation of a surface with a prescribed boundary and compatible diagonal curves}

Let us define another binary equivalence relation in the set of tensor product B\'ezier surfaces.

\begin{definition}
	Two tensor product B\'ezier surfaces of degree $n\times n$ are said to be boundary and diagonals equivalent if they have the same boundary and the same main diagonals curves.
\end{definition}

\begin{proposition}\label{boundarydiagonals}
	The class of boundary and diagonals equivalent surfaces is an affine subspace that can be parameterized by $(n-3)^2$ control points.
	
\end{proposition}

\begin{proof}
	Initially, for a degree $n\times n$ B\'ezier surface with a prescribed boundary we have $(n-1)^2$ free interior control points. 
	However, since the main diagonal curves must start and end at the corner control points of the net and moreover, the starting and ending derivatives of the diagonal curves are also prescribed by the boundary points, that is to say,
\begin{equation}\label{corners}
	\begin{array}{rclcccl}
		Q_0 &=& P_{0,0}&\quad& Q_1 &=& \frac12\left(P_{1,0}+P_{0,1}\right)\\[1mm]
		Q_{2n-1} &=& \frac12\left(P_{n-1,n}+P_{n,n-1}\right)&\quad& Q_{2n} &=& P_{n,n}\\[1mm]
		
		R_0 &=& P_{n,0}&\quad& R_1 &=& \frac12\left(P_{n-1,0}+P_{n,1}\right)\\[1mm]
		R_{2n-1} &=& \frac12\left(P_{1,n}+P_{0,n-1}\right)&\quad&
		R_{2n} &=& P_{0,n}\\[1mm]
	\end{array}
\end{equation}
	then, the system in Eq. (\ref{the-system}), has eight fewer equations,
	\begin{equation}\label{linear-system-boundary}
		\left\{
		\begin{array}{rcl}
			D^n_k({\mathcal P } )&=&  \binom{2n}{k} Q_k, \qquad k=2,\dots, 2n-2\\[2mm]
			E^n_k({\mathcal P } ) &=& \binom{2n}{k} R_k, \qquad k=2,\dots, 2n-2
		\end{array}.
		\right.
	\end{equation}
	We show in the proof of Proposition \ref{prop1}, that the rank of the system in Eq. (\ref{the-system}) is $4n$, so now the rank of the reduced system, in Eq. (\ref{linear-system-boundary}), is $4(n-2)$. Therefore, the number of free control points is $(n-1)^2 - 4(n-2) = (n-3)^2$.
	
\end{proof}

According to this result, for degrees $n=2$ and $n=3$, if we prescribe the boundary and the main diagonal curves, the B\'ezier surface is totally determined. The situation changes for $n\ge 4$. For example, for $n=4$ there is a family of B\'ezier surfaces, which is parameterized by just one interior control point, in such a way that any element of this family has the same prescribed boundary and diagonal curves.

Let us give some low degree examples in order to show the conditions that two control polygons ${\mathcal Q}$ and ${\mathcal R}$ must meet in order to be those of the diagonal curves of a B\'ezier surface with a prescribed boundary.

\begin{example}
	For $n=3$. 	When the boundary is prescribed the endpoints of the diagonal curves and the respective tangent lines at these endpoints are fixed, namely $Q_0, Q_1, Q_{5}, Q_{6}$ and $R_0,R_1, R_{5},R_{6}$, see Eq. (\ref{corners}). Therefore, the diagonal curves can be controlled by $Q_2,Q_4$ and $R_2,R_4$, which can be interpreted as meaning that the curvature at the endpoints remains free.
	Hence, we fix the boundary of the B\'ezier surface as follows:
	$$\left(
	\begin{array}{cccc}
		P_{0,0} & P_{0,1} & P_{0,2} & P_{0,3} \\[1mm]
		P_{1,0} & \frac{1}{3} \left(5 Q_2-P_{0,2}-P_{2,0}\right) & \frac{1}{3} \left(5 R_2-P_{0,1}-P_{2,3}\right) & P_{1,3} \\[1mm]
		P_{2,0} & \frac{1}{3} \left(5 R_4-P_{1,0}-P_{3,2}\right) & \frac{1}{3} \left(5 Q_4-P_{1,3}-P_{3,1}\right) & P_{2,3} \\[1mm]
		P_{3,0} & P_{3,1} & P_{3,2} & P_{3,3} \\
	\end{array}
	\right)
	$$
and using Corollary \ref{corollary-two-necessary-conditions} and the conditions at the corners, in Eq. (\ref{corners}), the control polygons of the main diagonals are
	$${\mathcal Q} = \left\{P_{0,0},\frac{1}{2} \left(P_{0,1}+P_{1,0}\right),Q_2,Q_3,Q_4,\frac{1}{2} \left(P_{2,3}+P_{3,2}\right),P_{3,3}\right\}$$
	$$
	{\mathcal R} =\left\{P_{0,3},\frac{1}{2} \left(P_{0,2}+P_{1,3}\right),R_2,R_3,R_4,\frac{1}{2} \left(P_{2,1}+P_{3,1}\right),P_{3,0}\right\},
	$$
	where
	$$
	\aligned &Q_3=\frac{3}{4}\left( R_2+ R_4\right)+\frac{1}{20} \left(P_{0,3}+P_{3,0}\right)-\frac{3}{20} \left(P_{0,1}+P_{1,0} +P_{2,3}+P_{3,2}\right)\\
	&R_3= \frac34(Q_2+Q_4)+\frac{1}{20} \left(P_{0,0}+P_{3,3}\right)-\frac{3}{20} \left(P_{0,2}+ P_{1,3}+ P_{2,0}+ P_{3,1}\right).
	\endaligned
	$$

Therefore, if we fix the boundary of a cubic B\'ezier curve and its (compatible) diagonals, then the surface is fully determined.
		
\end{example}

In all the following figures the control points that can be freely prescribed are represented by large points, while those that are determined (by the large ones) are represented by small points. The Bézier surface control net is shown in green and the diagonals control points in red. 

For example, in Figure \ref{fig4}, the prescribed boundary control points are, large green points, and the diagonal curves that can be controlled by $Q_2,Q_4$ and $R_2,R_4$, are large red points. The remaining control points in the net and in the diagonals polygons are determined and are represented by small points.
\begin{center}
\begin{figure}[h!]\centering
\includegraphics[width =5.7cm]{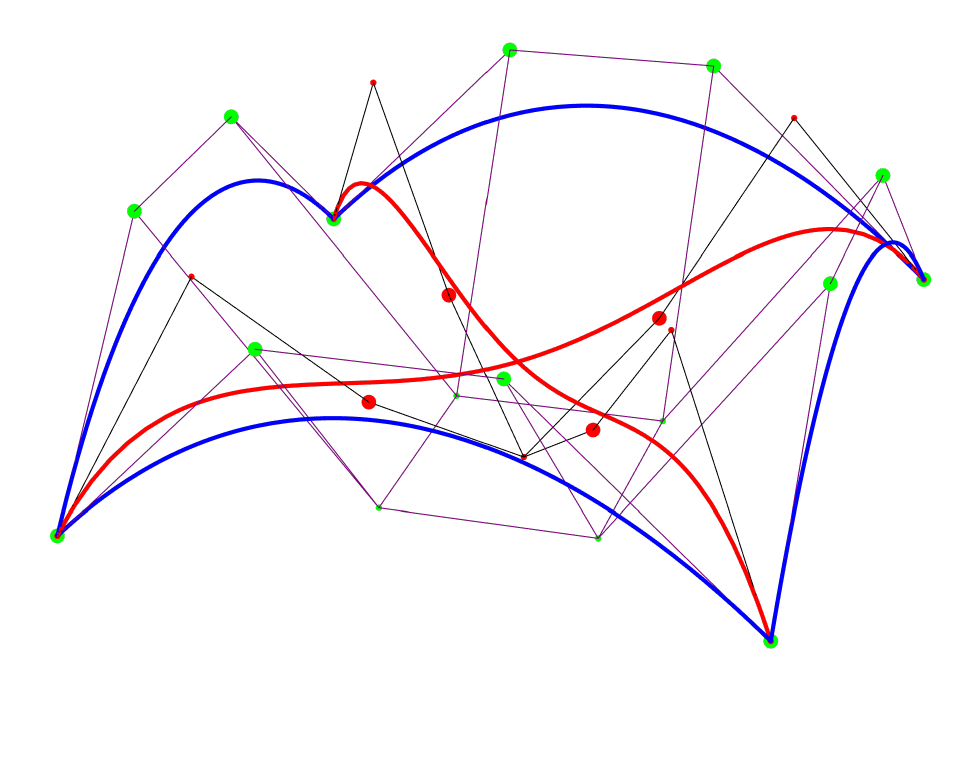}\qquad\qquad\includegraphics[width =5.8cm]{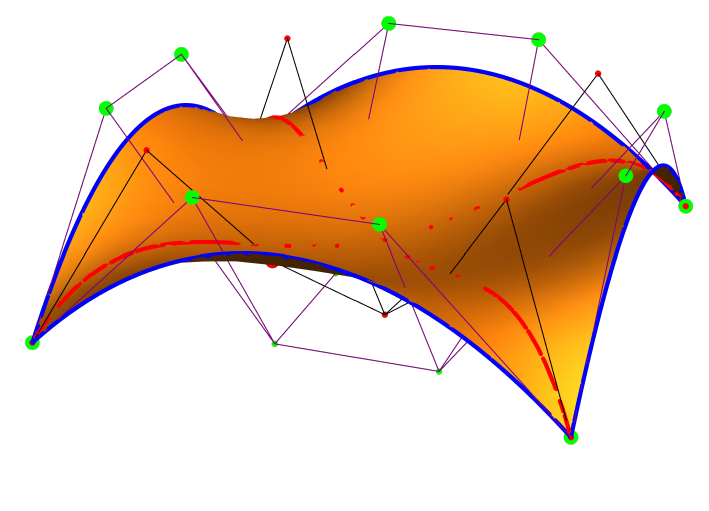}		
		\caption{$n=3$. Prescribed diagonal curves and boundary (properly related).}\label{fig4}
	\end{figure}
\end{center}

\begin{example}
	For $n=4$. If the boundary of the B\'ezier surface is prescribed, then the interior control points are
	
	$$\left(
	\begin{array}{ccc}
		\frac{7 Q_2}{4}-\frac{3}{8} \left(P_{0,2}+P_{2,0}\right) & \mathbf{P_{1,2}} & \frac{7 R_2}{4}-\frac{3}{8} \left(P_{0,2}+P_{2,4}\right) \\[1mm]
		\frac{1}{6} \left(-P_{0,3}-P_{3,0}+14 Q_3\right)- \mathbf{P_{1,2}} & P_{2,2}  &  P_{2,3}\\[1mm]
		\frac{7 R_6}{4}-\frac{3}{8} \left(P_{2,0}+P_{4,2}\right) & P_{3,2} & \frac{7 Q_6}{4}-\frac{3}{8} \left(P_{2,4}+P_{4,2}\right) \\
	\end{array}
	\right)$$
where
$$
\aligned 
&P_{2,2}=\frac{1}{36} \left(-P_{0,0}+6 P_{0,2}+6 P_{2,0}+6 P_{2,4}+6 P_{4,2}-P_{4,4}-28 Q_2-28 Q_6+70 R_4\right),\\
&P_{3,2}=\frac{1}{6} \left(P_{0,3}-P_{1,0}+P_{3,0}-P_{4,3}-14 Q_3+14 R_5\right)+\mathbf{P_{1,2}},\\
&P_{2,3}=\frac{1}{6} \left(-P_{0,3}+P_{1,0}-P_{1,4}-P_{3,0}-P_{4,1}+P_{4,3}+14 Q_3+14 Q_5-14 R_5\right)- \mathbf{P_{1,2}}.
\endaligned	
$$
	
Let us remark that, in this case, the diagonals are controlled by the eight control points $Q_2,Q_3,Q_5,Q_6,$ $R_2, R_4, R_5,R_6$ but, in addition, $P_{1,2}$ is free. The two diagonals do not depend on $P_{1,2}$. Indeed, the control polygons are
$$\begin{array}{lll}
		{\mathcal Q} &=\left\{P_{0,0},\frac{1}{2} \left(P_{0,1}+P_{1,0}\right),Q_2,Q_3,Q_4,Q_5,Q_6,\frac{1}{2} \left(P_{3,4}+P_{4,3}\right),P_{4,4}\right\},\\
		{\mathcal R} &=\left\{P_{0,4},\frac{1}{2} \left(P_{0,3}+P_{1,4}\right),R_2,R_3,R_4,R_5,R_6,\frac{1}{2} \left(P_{3,0}+P_{4,1}\right),P_{4,0}\right\},\end{array}
$$
with
	$$\aligned &Q_4= R_4+ \frac{2}{5}\left(R_2+ R_6- Q_2- Q_6 \right)+\frac{1}{70} \left(-P_{0,0}+P_{0,4}+P_{4,0}-P_{4,4}\right),\\
	&R_3=\frac{1}{14} \left(P_{0,1}-P_{0,3}+P_{1,0}-P_{1,4}-P_{3,0}+P_{3,4}-P_{4,1}+P_{4,3}\right)+ Q_3+ Q_5- R_5.
	\endaligned$$

\begin{center}	
\begin{figure}[h!]\centering
\includegraphics[width=4.5cm]{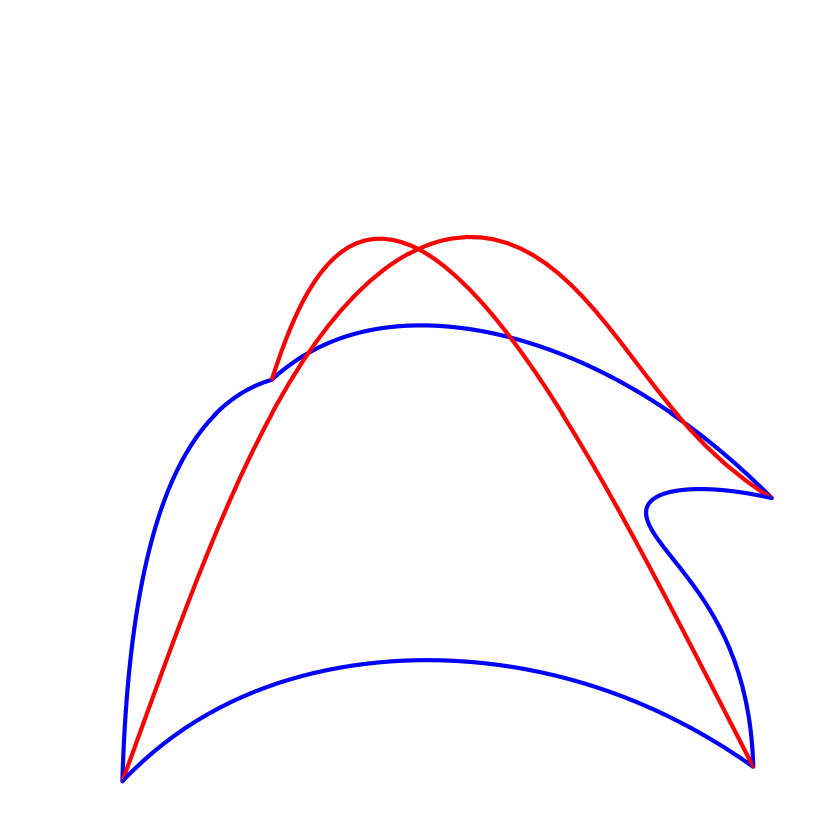}\qquad\quad\quad
		\includegraphics[width=4cm]{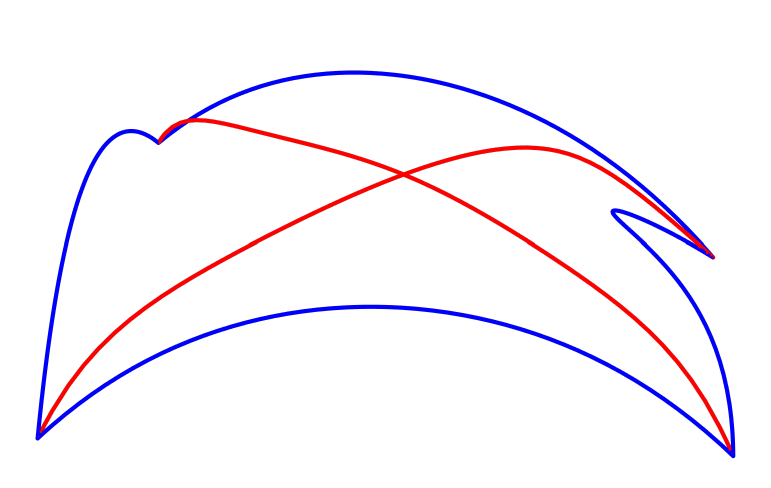}\\
\includegraphics[width =4.4cm]{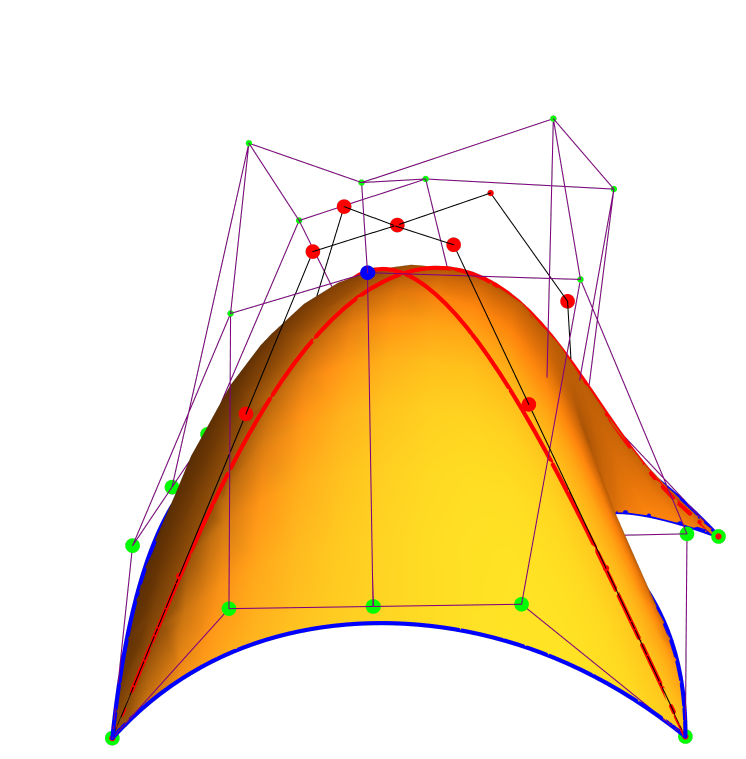}\qquad\quad\qquad\includegraphics[width =4cm]{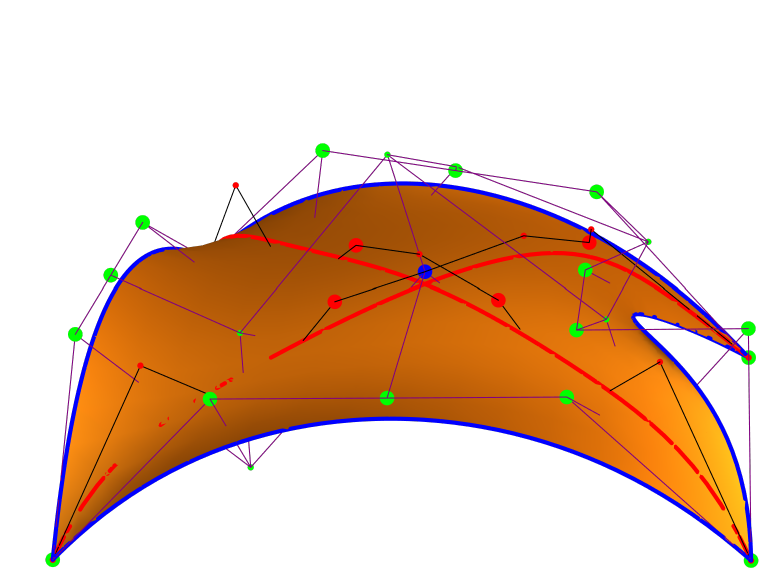}\\
\qquad\includegraphics[width=4.7cm]{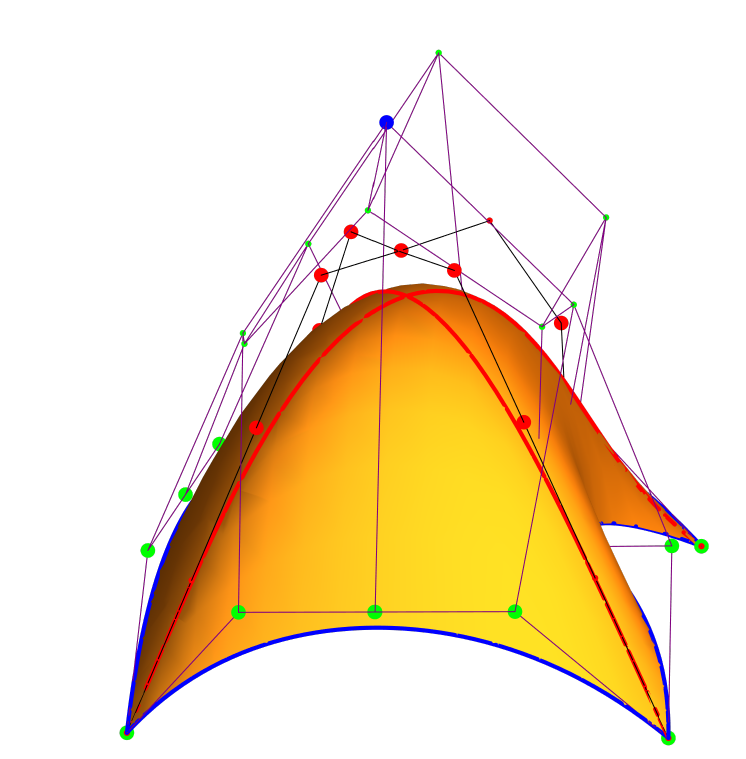}\qquad\quad\qquad\includegraphics[width=4cm]{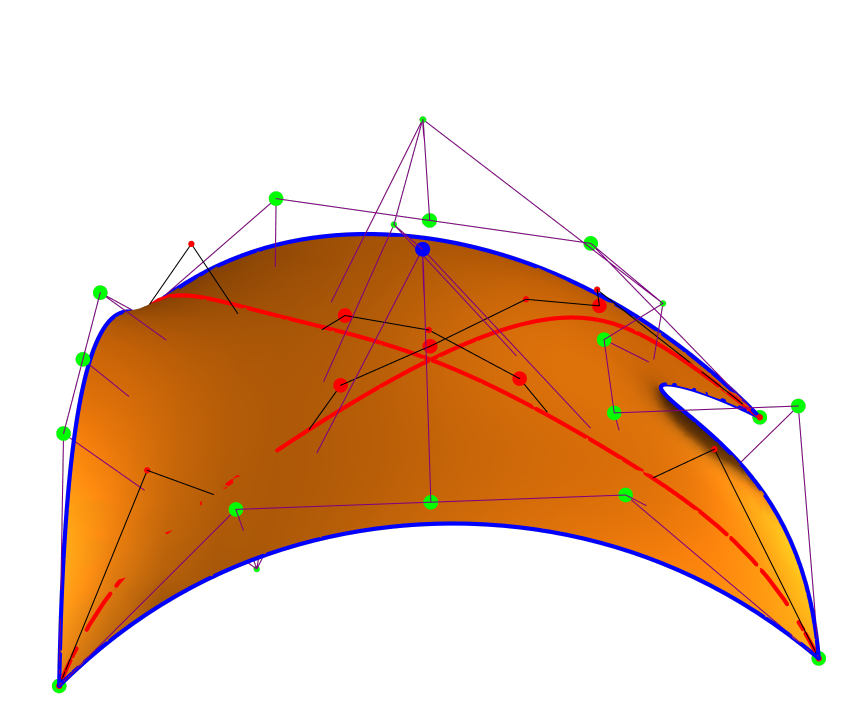}	
\caption{$n=4$. Same boundary curves for all figures but two different pairs of main diagonals. In each column, with the same diagonal curves, two surfaces are obtained by changing the free control point $P_{1,2}$, in blue.}
	\end{figure}
\end{center}		
\end{example}

\section{Generation of a surface with a prescribed boundary and tangent planes along it with compatible diagonal curves}

Let us define the following binary equivalence relation in the set of tensor product B\'ezier surfaces.

\begin{definition}
	Two tensor product B\'ezier surfaces of degree $n\times n$ are said to be $C^1$-boundary and diagonals equivalent if they have the same boundary, the same tangent planes along it and the same main diagonal curves.
\end{definition}

\begin{proposition}
		The class of  $C^1$-boundary and diagonals equivalent surfaces is an affine subspace that can be parameterized by $(n-5)^2$ control points.
\end{proposition}

\begin{proof} For a degree $n\times n$ B\'ezier surface with a prescribed $C^1$-boundary we have $(n-3)^2$ free interior control points, that is to say, the boundary control points and their neighboring lines in the net are prescribed. If we also prescribe the two main diagonals, these $(n-3)^2$ control points are related into a compatible linear system of rank $4(n-4)$, since $16$ equations from the initial system with rank  $4n$, in Eq. \ref{the-system},  only relate prescribed points. Therefore, the number of free control points is $(n-3)^2 - 4(n-4) = (n-5)^2$.
\end{proof}

According to this result and similarly to the previous work, if we prescribe the boundary, the tangent planes along it and the two diagonal curves, then the smallest degree surfaces with free control points are degree $6 \times 6$ surfaces. For $n=6$ there is a family of B\'ezier surfaces, parameterized by just one interior control point such that, any element of this family has the same prescribed boundary and tangent plane along it and the same prescribed diagonal curves.

As before, we give some low degree examples.

\begin{example}
	For $n=5$.	If we fix the boundary, the tangent planes along it and the main diagonals we determine the whole surface. Corollary \ref{corollary-two-necessary-conditions} determines the central control points, $Q_5$ and $R_5$, of the main diagonals and moreover $Q_0,Q_1,Q_2,Q_3,Q_7,Q_8,Q_9,Q_{10}$ and $R_0,R_1,R_2,R_3,R_7,R_8,R_9,R_{10}$  depend on prescribed points of the control net, see Eq. (\ref{QR_P}).
	The interior $(n-3)^2$ control points of the net are as follows:
	$$
	\left(
	\begin{array}{cc}
		\frac{1}{20} \left(42 Q_4-P_{0,4}-10 P_{1,3}-10 P_{3,1}-P_{4,0}\right) & \frac{1}{20} \left(42 R_4-P_{0,1}-10 P_{1,2}-10 P_{3,4}-P_{4,5}\right) \\
		\frac{1}{20} \left(42 R_6-P_{1,0}-10 P_{2,1}-10 P_{4,3}-P_{5,4}\right) & \frac{1}{20} \left(42 Q_6-P_{1,5}-10 P_{2,4}-10 P_{4,2}-P_{5,1}\right) \\
	\end{array}
	\right),
	$$
where $Q_4, Q_6$ and $R_4, R_6$ are the way of controlling the shape of the diagonal curves.
	
	\begin{figure}[h!]\centering
\includegraphics[width =5cm]{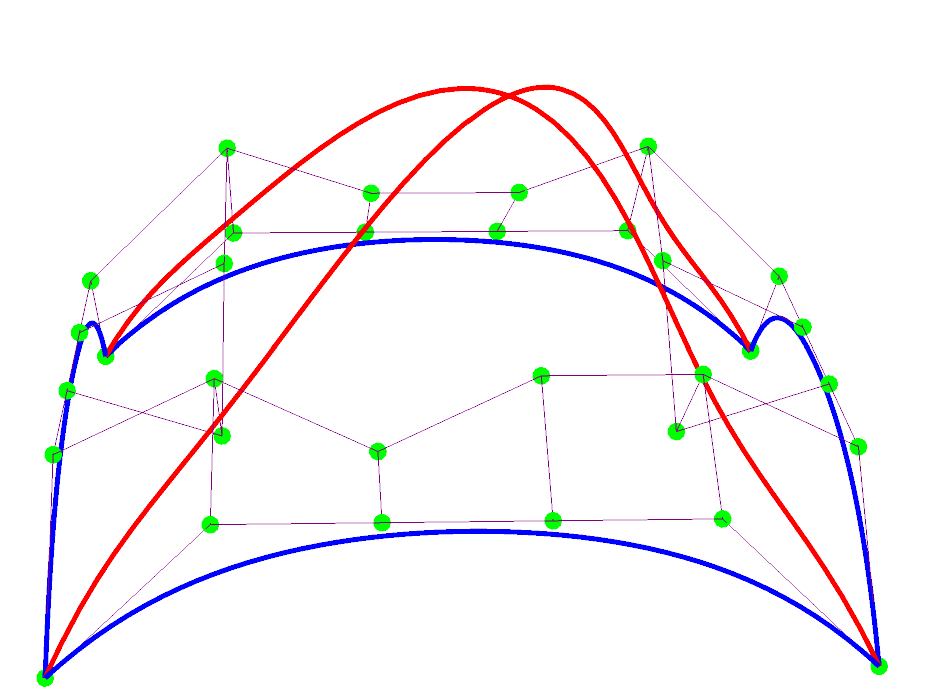}\quad\quad\quad\quad \includegraphics[width=5cm]{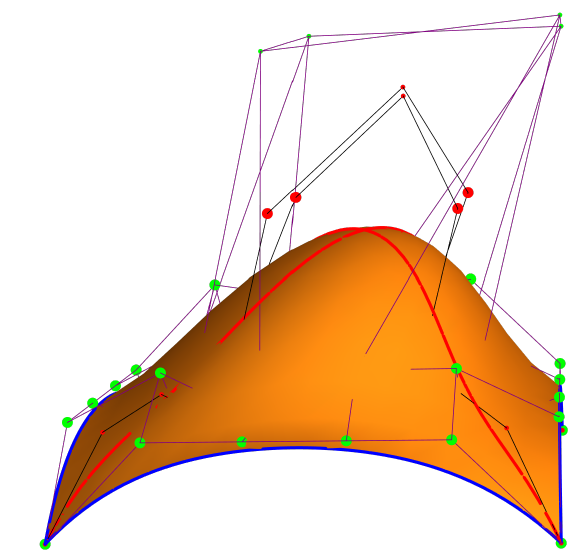}
\caption{$n=5$. Prescribed boundary and tangent planes along the boundary, shown as large green points, and the control points $Q_4, Q_6$ and $R_4,  R_6$, shown as large red points, let us control the shape of the diagonals. Again after prescribing the $C^1$-boundary control points and the diagonal curves the whole surface is determined.}\label{c15}
	\end{figure}
\end{example}

\begin{example}
	For $n=6$.	There is a one parameter family of surfaces with the same boundary and tangent plane along it and with the same diagonal curves. In other words, the subspace of  $C^1$-boundary and diagonals equivalent $6\times6$ Bézier surfaces can be parameterized by one control point, meaning that the dimension of the subspace is one.  The interior $(n-3)^2$ control points are
	$$
	\aligned
	&P_{2,2}=\frac{1}{15} \left(33 Q_4-P_{0,4}-8 P_{1,3}-8 P_{3,1}-P_{4,0}\right), \\
	&P_{2,4}= \frac{1}{15} \left(33 R_4-P_{0,2}-8 P_{1,3}-8 P_{3,5}-P_{4,6}\right), \\
	&P_{3,2}=
	\frac{1}{50} \left(132 Q_5-P_{0,5}-15 P_{1,4}-15 P_{4,1}-P_{5,0}\right)-{\bf P_{2,3}} ,\\
	&P_{3,3}=  \frac{1}{400} \left(924 R_6-495 Q_4-495 Q_8-P_{0,0}+15 P_{0,4}-36 P_{1,1}+120 P_{1,3}+15 P_{2,6}\right.\\
	&\left.\qquad+120 P_{3,1}+120 P_{3,5}+15 P_{4,0}+120 P_{5,3}-36 P_{5,5}+15 P_{6,2}-P_{6,6}\right), \\
	&P_{3,4}=  \frac{1}{50} \left(132 R_5-P_{0,1}-15 P_{1,2}-15 P_{4,5}-P_{5,6}\right)-{\bf P_{2,3}},\\
	&P_{4,2}=
	\frac{1}{15} \left(33 R_8-P_{2,0}-8 P_{3,1}-8 P_{5,3}-P_{6,4}\right), \\
	&P_{4,3}= {\bf P_{2,3}}+ \frac{1}{50} \left(132 Q_7-132 R_5+P_{0,1}+15 P_{1,2}-P_{1,6}-15 P_{2,5}+15 P_{4,5}\right.\\
	&\left.\qquad	-15 P_{5,2}+P_{5,6}-P_{6,1}\right), \\
	&P_{4,4}=  \frac{1}{15} \left(33 Q_8-P_{2,6}-8 P_{3,5}-8 P_{5,3}-P_{6,2}\right), \\
	\endaligned
	$$
which depend on the free parameter $P_{2,3}$, $Q_4, Q_5, Q_7, Q_8$ and $R_4, R_5, R_6,  R_8$ are the diagonal curves control points that can be prescribed.
	
\begin{center}	\begin{figure}[h!]
	
\includegraphics[width =3.5cm]{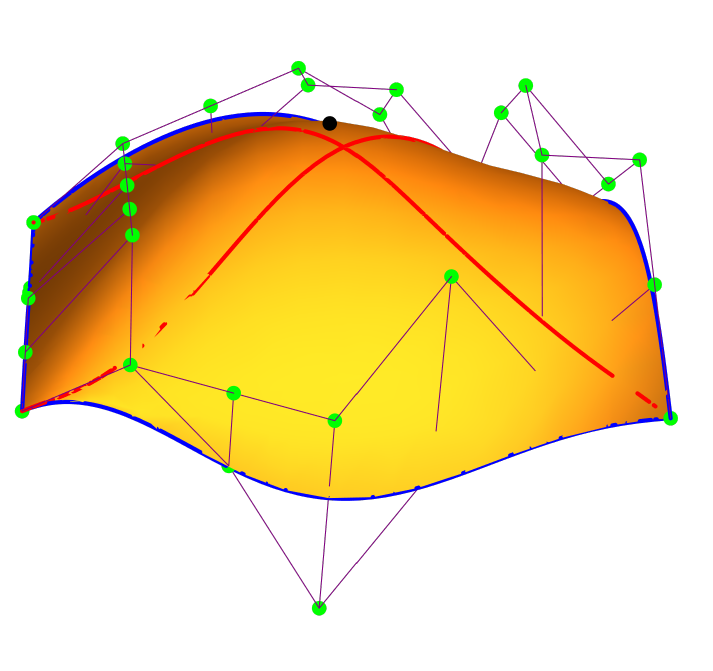}\quad\includegraphics[width =4.25cm]{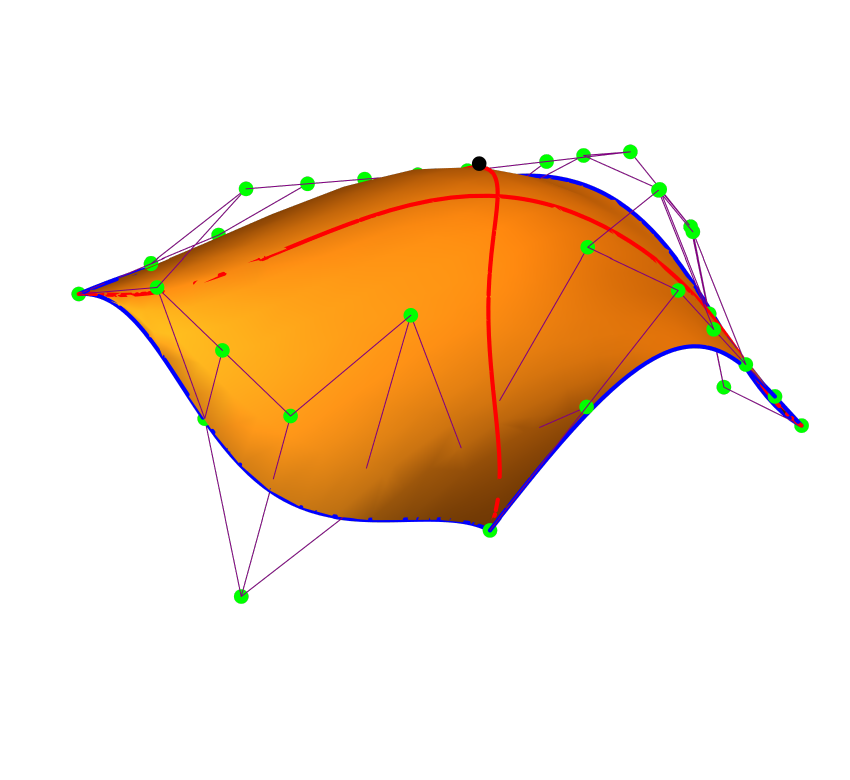}\quad\includegraphics[width =4.25cm]{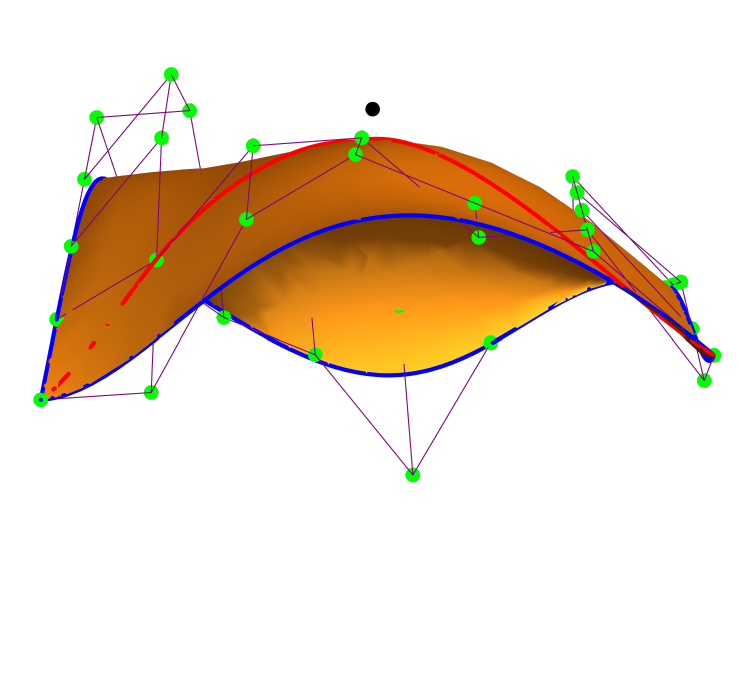}
		\centering
		\caption{$n=6$. Given a boundary, the tangent plane along it and the diagonal curves of a Bézier surface, $Q_4, Q_5, Q_7, Q_8$ and $R_4, R_5, R_6, R_8$ being the diagonal curves control points that can be prescribed, there is still one free control point, $P_{2,3}$ (black). }\label{c16}
	\end{figure}
\end{center}

\end{example}

After studing all these different methods for generating surfaces from prescribed information let us make only a comment about multi-patch surfaces. Two patches could be, of course, continously connected along matching boundaries, and they could also be $C^1$ connected when we prescribe tangent planes along the boundaries of neighbour patches and they agree. But the thing is that, since the tangent lines at the diagonal curves at the corner points depend on the boundary control points, see Eq. (\ref{QR_P}), if certain continuity conditions on the diagonal curves are wanted, it implies that conditions on the prescribed boundary and tangent planes must be imposed.  A detailed study of this issue, but out of our scope here, shows that if we want the diagonals to be well connected, then the connection of the patches at the common point to four patches must be $C^1$.

In short, if we consider the problem of multi-patches and we want a good connection of diagonal curves, this entails some conditions on the boundary and the tangent planes.

\section{Conclusion}

First, in this paper, we study the necessary and sufficient conditions that two curves must verify in order to constitute the diagonal curves of a tensor product B\'ezier surface, which allows us to determine tensor product B\'ezier surfaces with prescribed diagonal curves.

The possibility of controlling the diagonals would increase the control of the resulting surface, so, the goal we achieve here is to provide a method that makes it possible to prescribe not only the diagonal curves of the surface but also its boundary, or its boundary and the tangent planes along it.

Whereas for low degrees, it is easy to check that there is only one B\'ezier surface with prescribed main diagonals and boundary, this is no longer true for higher degrees.  For fixed diagonals and boundary not a unique Bézier surface but a family of Bézier surfaces is determined. Such a family has the structure of an affine subspace whose dimension is also determined here.

In other words, there is an equivalence class of B\'ezier surfaces that have the same diagonals, or the same diagonals and boundary as well as having the same diagonals and meeting the same $C^1$-boundary conditions.

In fact, what we are able to prove is that the set of all degree $n\times n$ B\'ezier surfaces where both the diagonal curves and the boundary are prescribed is an affine space of dimension $(n-3)^2$ whereas the set of surfaces with prescribed diagonals and $C^1$-boundary conditions has the dimension $(n-5)^2$. Therefore, although only for degrees $n\le 3$ and $n\le5$, respectively, a unique B\'ezier surface exists as a solution. For higher degrees, there are free control points, that could also be determined by functional minimization, considering functionals defined on the surface, such as the Dirichlet functional.

\section*{Acknowledgements}
This work was partially supported by Grant PGC2018-094889-B-100 funded by MCIN/AEI/ 10.13039/\newline 501100011033 and by ``ERDF A way of making Europe''.
The first author has  also been funded by Ministerio de Ciencia e Innovaci\'on  (Spain) through project PID2019-104927GB-C21, and by Universitat Jaume I (grants UJI-B2019-17 and GACUJI/2020/05). 

We thank the referees for their careful reading of our manuscript and for taking the time to comment on it and offering constructive suggestions.

\end{document}